\documentclass[submission,copyright,creativecommons]{eptcs}
\providecommand{\event}{SNR 2017} % Name of the event you are submitting to
\usepackage{breakurl}             % Not needed if you use pdflatex only.
\usepackage{underscore}           % Only needed if you use pdflatex.
\usepackage[utf8]{inputenc} % set input encoding (not needed with XeLaTeX)
\usepackage{amsmath,tabu}
\usepackage{amssymb}
\usepackage{float}
\usepackage{alltt}
\usepackage{proof}
\usepackage{courier}
\usepackage{fancyvrb}
\usepackage{verbatim}
\usepackage{mdframed}
\usepackage{nicefrac}
\usepackage{pgfplots} \pgfplotsset{compat=newest}
\usepackage{semantic}
\usepackage{tikz}
\usepackage{amsthm}
\usepackage{wrapfig}
\usepackage{enumitem}
\usepackage{multirow}
\usepackage{xcolor}% http://ctan.org/pkg/xcolor
\usepackage{colortbl}
\usepackage{subfig}

\newcommand\offset{0.4in}
\newcommand*{\rom}[1]{\expandafter\@slowromancap\romannumeral #1@}
\newcommand\CoreDEL{Acumen-17}

\newcommand\labelE{\Delta}
\newcommand\bool{t}

\newcommand{\vb}[1][]{%
\ifthenelse{\equal{#1}{}}{v^b}{v^{#1}}%
}
\newcommand{\tf}{D_f}
\newcommand{\pf}{P_f}
\newcommand{\sbs}{s^{b_s}}

\newcommand{\dones}{w^{b_s}}
\newcommand{\vtwo}[2]{v_{#1}^{#2}}

\newcommand\wb{w^b}

\newcommand{\naturals}[0]{\ensuremath{\mathbb{N}}}

\newcommand{\domain}[1]{dom_{#1}}
\newcommand{\labels}[0]{\text{Label}}
\newcommand{\BT}[0]{\text{Binding time}}

\newcommand{\BTC}[0]{\text{Constraint}}

\newcommand\scope{l_0}
\newcommand\nlt{\sqsubseteq}

\newcommand\stat{\tb{S}}
\newcommand\dyn{\tb{D}}

\newcommand\ereduce{\hookrightarrow_{e}}
\newcommand\sreduce{\hookrightarrow_{s}}

\newcommand\btenv{\Gamma^b}
\newcommand\btenvd{\Gamma^{\dyn}}
\newcommand\tenv{\Gamma}

\newcommand\se{\rhd}

\newcommand\sigmahat{\tilde{\sigma}}
\newcommand\twosubs[2]{#1 \se #2}
\newcommand\twosigs{\twosubs{\sigma}{\sigmahat}}

\newcommand\err{\tb{err}}

\newcommand\la{\langle}
\newcommand\ra{\rangle}

\newcommand\tb[1]{\textsf{#1}}

\newcommand\type{\tau}
\newcommand\bt{b}
\newcommand\btt[1]{\type^{#1}}
\newcommand\vtype[1]{\prod\limits_{j \in 1..m}#1}
\newcommand\vtypef[1]{\prod\limits_{j \in 1..m_f}#1}
\newcommand\vbtype[2]{({\prod\limits_{j\in1..m}#2})^{#1}}
\newcommand\vbtypef[2]{({\prod\limits_{j\in1..m_f}#2})^{#1}}

\newcommand\const{k}
\newcommand\f{f}

\newcommand{\indexing}{\ensuremath{{j \in 1...m}}}
\newcommand\sequence[1]{\langle#1\rangle^{\indexing}}
\newcommand\vectoromit[1]{\langle#1\rangle}
\newcommand\sets[1]{\{#1\}^{\indexing}}
\newcommand\setsomit[1]{\{#1\}}
\newcommand\der[1]{\frac{d}{dt} #1}
\newcommand\pder[2]{\frac{\partial}{\partial #2} #1}

\newcommand{\setof}[1]{\ensuremath{\{ {#1} \}}}
\newcommand{\nequiv}{\not\equiv}
\newcommand{\done}{w^b}

\newcommand{\RNum}[1]{\uppercase\expandafter{\romannumeral #1\relax}}

\newcommand\scopeenv{\pi}

\newcommand{\labelSet}[1]{\ensuremath{Labels({#1})}}
\newcommand{\set}[1]{\ensuremath{\{{#1}\}}}
\newcommand{\setbar}[2]{\ensuremath{\{{#1}\;|\; {#2}\}}}
\newcommand{\CSets}[0]{\ensuremath{{\mathcal{C}}}}
\newcommand{\CNF}[0]{\ensuremath{NF}}
\newcommand{\CError}[0]{\ensuremath{Error}}
\usepackage{array}
\newcolumntype{L}[1]{>{\raggedright\let\newline\\\arraybackslash\hspace{0pt}}m{#1}}
\newcolumntype{C}[1]{>{\centering\let\newline\\\arraybackslash\hspace{0pt}}m{#1}}
\newcolumntype{R}[1]{>{\raggedleft\let\newline\\\arraybackslash\hspace{0pt}}m{#1}}

\title{Compile-Time Extensions to Hybrid ODEs}
\author{Yingfu Zeng
\institute{Rice University, Texas, USA}
\email{yz39@rice.edu}
\and
Ferenc Bartha
\institute{Rice University, Texas, USA}
\email{Ferenc.A.Bartha@rice.edu}
\and
Walid Taha
\institute{Halmstad University, Sweden}
\email{Walid.taha@hh.se}
}
\def\titlerunning{Compile-Time Extensions to Hybrid ODEs}
\def\authorrunning{Yingfu Zeng, Ferenc Bartha \& Walid Taha}

\begin{document}
\maketitle

\begin{abstract}
Reachability analysis for hybrid systems is an active area of development and has resulted in many promising prototype tools. Most of these tools allow users to express hybrid system as automata with a set of ordinary differential equations (ODEs) associated with each state, as well as rules for transitions between states. Significant effort goes into developing and verifying and correctly implementing those tools. As such, it is desirable to expand the scope of applicability tools of such as far as possible. With this goal, we show how compile-time transformations can be used to extend the basic hybrid ODE formalism traditionally supported in hybrid reachability tools such as SpaceEx or Flow*. The extension supports certain types of partial derivatives and equational constraints. These extensions allow users to express, among other things, the Euler-Lagrangian equation, and to capture practically relevant constraints that arise naturally in mechanical systems. Achieving this level of expressiveness requires using a binding time-analysis (BTA), program differentiation, symbolic Gaussian elimination, and abstract interpretation using interval analysis. Except for BTA, the other components are either readily available or can be easily added to most reachability tools. The paper therefore focuses on presenting both the declarative and algorithmic specifications for the BTA phase, and establishes the soundness of the algorithmic specifications with respect to the declarative one.
\end{abstract}

\section{Introduction}

Reachability analysis for hybrid systems \cite{Alur93} is an active area of development and has resulted in many promising prototype tools. Prominent examples of such tools include CHARON \cite{CHARON}, HyTech \cite{HyTech}, PHAVer \cite{phaver}, dReach \cite{dreach}, dReal \cite{dreal}, SpaceEx \cite{spaceex}, and Flow*\cite{flowstar}. Most of these tools allow users to express hybrid systems as automata with a set of ordinary differential equations (ODEs) associated with each state, as well as rules for transitions between states. In particular, ODEs must be in the explicit form where the left hand side of an equality has to be the derivative of a state variable. Significant effort goes into verifying and correctly implementing those tools. As such, it is desirable to expand the scope of applicability tools of such as far as possible.

\subsection{Contributions}

With this goal, we present a systematic method to translate an expressive language with partial derivatives and equations to a standard language supporting ODEs, guards, and reset maps. The method can be used to extend reachability analysis tool such as SpaceEx or Flow*.  An experimental implementation of the proposed technique is available in the freely available, open source Acumen language implementation \cite{acumenURL}. Examples illustrating the use of these extension can be found in the directory \verb|examples/04_Experimental/04_BTA|. Since both partial derivatives and equations are eliminated completely after the compile-time transformation,  the user benefits from the added expressivity but the underlying tools do not need to change. The two extensions allow the user to express, among other things, the Euler-Lagrangian equation, and to capture practically relevant constraints that arise naturally in mechanical systems. Achieving this level of expressivity requires using a binding time-analysis (BTA) \cite{Jones85,Gomard91,christensenaccurate,Moggi97}, program differentiation, symbolic Gaussian elimination, and abstract interpretation using interval analysis. Except for BTA, the other components are either readily available or can be easily added. The technical part of the paper therefore focuses on presenting both the declarative and algorithmic specifications for the BTA phase, and establishes the soundness of the algorithmic with respect to the declarative.

After reviewing related work on compile-time extensions (Section \ref{section:relatedwork}), we introduce the syntax and type system for a core differential equation language (Section \ref{section:syntax}). Then, we present a declarative specification of binding-time analysis (BTA) and a big step semantics for specialization (Section \ref{section:BTASpecialization}), along with a formal proof of type safety (Theorem 1).  We then present an algorithmic specification of the BTA that works by first generating a set of constraints and then attempting to solve them (Section \ref{section:Implementation}), and we show that this algorithmic specification is faithful to the declarative BTA (Lemma~\ref{lemma:faithful}) and always produces a unique  minimum solution that maps as much of the code as possible to static if an assignment exists (Theorem 2).  To illustrate the practical value of the formalism, we present two case studies that have been carried out using the implementation (Appendix \ref{appendix}).

\section{Related Work}\label{section:relatedwork}

Binding-time Analysis (BTA) is a static analysis traditionally supported in the offline partial evaluation of general purposes languages. It works by identifying a two-level structure in the program being analyzed, where the first level is a computation that can be done at ``partial evaluation time" (``compile time" in our case), and the second level must be left as a ``residual" that is executed at runtime. BTA has generally been studied for general purpose languages. In our setting, we study it in the context of Domain Specific Languages (DSLs) \cite{Hudak97,Mernik05,Sujeeth14} intended for modeling hybrid systems. It should also be noted that our primary purpose is to use it for extending expressivity. Partial evaluation, in contrast, is only concerned with improving the runtime performance of programs. In what follows, we elaborate on these key points.

A key idea in the work we present in this paper is that there are powerful techniques from the programming languages community that can help make reachability tools more broadly applicable. To do this, this work uses two-level languages in a novel way. To put the existing related work in context, it is useful to consider several characteristics relating to the language considered and the transformation used, namely, whether the language is domain-specific (or general purpose), whether it supports equations, whether the transformation is done at compile-time (or runtime), whether the tool performs let-insertion (to avoid code duplication), whether the language is statically typed, and whether the tool provide accurate source level error reporting.  The systems that we will consider are partial evaluation systems for C, namely, C-mix \cite{CMix} and Tempo \cite{tempo}; template  instantiation mechanisms, namely, C++ Templates \cite{metaOCamel} and Template Haskell \cite{templateHaskell}; multi-stage programming languages, namely, MetaOCaml \cite{metaOCamel} and LMS \cite{LMS}; the Verilog Preprocessor \cite{VPP}; and the hybridization technique \cite{bak}.
\begin{table*}
 \caption{Comparison of Compile-time Transformation }
\begin{center}
\small
  \begin{tabular}{ | p{20mm} | p{15mm} | p{20mm} | p{13mm} |  p{13mm}| p{20mm} |p{14mm} |}
    \hline
    &Static \newline checking&Source level \newline reporting  & Compile-time   & Domain-specific     &Let \newline insertion	         &   Equations \\ \hline
      C-Mix                     &Yes                         &-                    &Yes   			 & No	          &Yes				       &No                           \\ \hline
      Tempo                    &Yes                         &Yes                    &Yes   			   &No 		  &Yes				        &No                           \\ \hline
       C++ Template         &No                         &No                      &Yes    			    &No               &-					&No                           \\ \hline
      Template Haskell    &Yes                         &Yes                       &Yes    			     &No		  &-				        &No                           \\ \hline
      MetaOcaml             &Yes                         &Yes                      &No  			     &No 	          &No				        &No                           \\ \hline
      LMS                     &Yes                          &-                      &No   			   &No 	          &Yes					&No                           \\ \hline
      Verilog preprocessor     &Yes                       &Yes                      &Yes    			&Yes 	          &-					&No                          \\ \hline
      Hybridization          &Yes                         & -                      &Yes    			 & Yes	          &-         				&No                           \\ \hline
      This paper               &Yes                         & Yes                       &Yes    			 &Yes  			  &Yes 						&Yes                            \\ \hline
    
    \hline
  \end{tabular}
\end{center}

\label{fig:compare}
\end{table*}

Table \ref{fig:compare} provides an overview of how these different systems related to these key properties. The main observations from the table are as follows.  Almost all tools are compile-time (except MetaOCaml), and almost all are statically checked (except C++ Templates). A key feature of static checking is that it facilitates accurate source-level reporting. That is primary reason for choosing an approach based on BTA or some type of static analysis. Compile-time program specializers, such as C-Mix and Tempo focus on automatically specializing a program through a well understood set of transformations to produce a program that is faster than the original one. There are no fundamental reasons why specialization (and two-level languages) need to be limited to general purpose languages. In fact, as this paper shows, they can be quite useful as they can be used to increase expressivity. Let-insertion was invented in the partial evaluation community, and is adopted by automated tools by LMS (but not be explicit tools like MetaOCaml). It is quite critical when there are significant compile-time computations, as is the case when we are trying to eliminate non-trivial constructs like partial derivatives or performing substantial manipulations to turn equations into formulae.  However, none of these works address the question of supporting equations, that is, allowing the user to write constraints in equational form, and then translating them directly into ``formula" form for directed evaluation. 
\begin{figure*}[h]
\footnotesize
%\centering
\begin{tabular}{  l >{$}l<{$} >{$}c<{$} >{$}l<{$} l  }% Automatically into "Math mode" in this column
\multicolumn{3}{l}{\fbox{Syntax}}\\
\multicolumn{4}{l}{$ \hspace{1in}n \in \text{Names},\quad i \in \mathbb{N}, \quad q \in \mathbb{Q} \quad \text{and} \quad \bool\  \in \mathbb{B}$}\\
\hspace{\offset}$\text{Constant}$ &k &::= & i \ | \ q \ | \ \bool& \\
\hspace{\offset}$\text{Variable}$ &x &::= & n \ | \ x'& \\ 
\hspace{\offset}\text{Type}              &\type& ::=&  \tb{nat} \ | \ \tb{bool} \ | \ \tb{real}\  | \ \vtype{\type_j} \\
\hspace{\offset}\text{Type Environment} &\tenv& =&  \sets{x_j:\type_j} \ \text{such that whenever} \ x_i = x_j \ \text{then} \ i = j  \\
\hspace{\offset}$\text{Function}$ &f &::= & \tb{+} \ | \  \tb{-} \ | \ \tb{$\times$} \ | \  \tb{/} \ | \ \tb{\^{}} \ | \  \tb{\&\&} \ | \  \textsf{$||$} \ | \  \textsf{$>$} \ | \  \textsf{$>=$}
						\ | \  \tb{==}\ | \  \tb{!=}\ | \  \sin\ | \  \cos& \\
\hspace{\offset}$\text{Expression}$&e&::=& \const\  |\  x \ |\  \sequence{e_j} \ | \  e_1(e_2) \ | \ \f(e) \ | \ \der{e} \ | \ \pder{e_1}{e_2}  & \\
\hspace{\offset}$\text{Equation}$&s &::=& x = e \ | \ e_1 = e_2 ,\ e_1 \nequiv x \ | \ x^{+} = e  \ | \ 
                                           \tb{if} \ e\ \tb{then} \ s_1 \ \tb{else} \ s_2 \ | \ \forall n \in e. \ s \ | \ \sets{s_j}&\\
%&&&\\
\multicolumn{3}{l}{\fbox{$\tenv \vdash e:\type$}}\\       
\multicolumn{5}{c}{      
$\begin{array}{ c }
\infer{\Gamma \vdash i: \tb{nat}} {} \quad \infer{\Gamma \vdash q: \tb{real}} {} \quad \infer{\Gamma \vdash t: \tb{bool}} {} \quad
\quad \infer{\Gamma \vdash n : \type} { n: \type \in \tenv} \quad
 \infer{\Gamma \vdash x' : \tb{real}} { x': \tb{real} \in \tenv & x :\tb{real} \in \tenv} \quad
\infer{\Gamma \vdash \sequence{e_j}:\vtype{\type_j}\ } {\Gamma \vdash e_j : \type_j} \quad
\\
\\
\infer{\Gamma \vdash e_1(i): \type_i } {\Gamma \vdash e_1 : \vtype{\type_j}\\ \Gamma \vdash i:\tb{nat}\ \ i<m} \quad
\infer{\Gamma \vdash e_1(e_2): \type } {\Gamma \vdash e_1 : \vtype{\type}\\ \Gamma \vdash e_2:\tb{nat}}\quad 
\infer{\Gamma \vdash \f(e):\type} 
        {\multicolumn{1}{l}{\Gamma \vdash e : \vtypef{\type_{f,j}}} \\ \multicolumn{1}{l}{\f : \vtypef{\type_{f,j}} -> \type}} \quad
\infer{\Gamma \vdash \der{e} : \tb{real}} { \Gamma \vdash e:\tb{real}}\quad
\infer{\Gamma \vdash \pder{e_1}{e_2} : \tb{real}} 
	  { \Gamma \vdash e_1:\tb{real} \ \  \Gamma \vdash e_2:\tb{real}}\\                 
\end{array}$}\\ 
\multicolumn{1}{l}{\fbox{$\tenv \vdash s$}}
\end{tabular}
\[
\begin{array}{c}
%\infer{\Gamma \vdash x = e } { \Gamma(x) : \type & \Gamma \vdash e:\type}\quad
\infer{\Gamma \vdash e_1 = e_2 } { \Gamma \vdash e_1 : \type \ \Gamma \vdash e_2:\type } \quad   
\infer{\Gamma \vdash x^{+} = e } { \Gamma \vdash x: \type & \Gamma \vdash e:\type}\quad
\infer{\Gamma \vdash \tb{if} \ e\ \tb{then} \ s_1 \ \tb{else} \ s_2 }
	  {\Gamma \vdash e : \tb{bool} \ \Gamma \vdash s_1  \  \Gamma \vdash s_2  \  }\quad
\infer{\Gamma \vdash \forall n \in e \ . \ s}
	  { \Gamma \vdash e:\vtype{\type}  \quad   \tenv, n:\type \vdash s }\quad
\infer{\Gamma \vdash \sets{s_j} } 
        { \Gamma \vdash s_j }\\
\end{array}
\]
\caption{Syntax and Type system for \CoreDEL}
\label{fig:Lsyntaxtype}
\end{figure*}
Our work is comparable to that of HyST \cite{hyst}, which is a tool that aims
to facilitate interchange of models between different tools. This way, HyST facilitates sharing of models and comparing solvers. In contrast, our work explores another dimension for reuse, namely, how these tools can be extended to support a more flexible and expressive modeling formalism.

\section{A Differential Equation Language (\CoreDEL)}\label{section:syntax}
Fig.~\ref{fig:Lsyntaxtype} introduces the syntax and type system for a core differential equation language called \CoreDEL.  We use the
following notational conventions:
\begin{enumerate}
	\item[-] Writing $\sequence{e_j}$ denotes a vector $\la e_1,
         e_2,...,e_m\ra$.  We will occasionally omit the superscript
          $j \in \{1...m\}$ and write $\la e_j \ra$ when the range of
          $j$ is clear from context.
	 \item[-] Writing $\sets{e_j}$ denotes a set
           $\{e_1,e_2,...,e_m\}$, and we write $A \uplus B$ for $A
           \cup B$ when we require that $A \cap B = \emptyset$.
\end{enumerate}

The set $\text{Names}$ is a finite countable set of names, and we
use $n$ to denote elements of this set.  We use $i$ to range over
natural numbers, $q$ to range over rationals, and $t$ to range over
booleans.

Similarly, we introduce the natural number $i$ drawn from the set of
natural numbers $\mathbb{N} $, rational $q$ from rational number set
$\mathbb{Q} $ and lastly boolean values $t$ from $\set{\tb{true},
  \tb{false}}$, denoted by $\mathbb{B}$. Variables are either a name
$n$ or a name followed by a number of primes ($'$).
Type terms represent naturals, reals, Booleans, and products,
respectively.  A type environment is a partial function from variables
to type terms.  We treat environments as graphs of functions or as
functions.

Function names $f$ are drawn from a fixed set containing basic
operators.  Expressions include constants, variables, vector
expressions, vector indexing, function application, time derivatives,
and partial derivatives.  Derivatives can be applied to both expressions and variables. The time derivative on a variable, for example $x''$, has special status, in that it can both be used in expressions to mean the value of the derivative at a given time and can also be equated to a value. When there is a constraint that equates a time derivative of a variable to a value, the effect is that integration is used to compute the value of the variable itself.  In principle, in an equational language, if a symbolic expression for the variable is known, the derivative variable can be determined from that expression.  In practice, it is generally rare that a closed form expression for the result of a simulation is known. Instead, it is more common to have the value of the derivative known, and then numerical integration is used to compute the value of the variable itself.  
The partial derivative ($\pder{e_1}{e_2}$) is an operator that takes two expressions and returns the result of the first expression differentiated with respect to the second expression. The  ASC\RNum{2}-based syntax is $\verb|expr'[expr]|$.  For two scalars, the result is simply the first expression partially differentiated with the second one. If one expression is a scalar, and the other a vector, the operator is applied component-wise.  Allowing arbitrary expressions $e_2$ instead
of just variables in partial derivatives, allows us to express things
like the Euler-Lagrange equation directly. 

The first type of equations is a continuous equation.  In processing
such equations, we distinguish between two cases, one where the left
hand side is a variable, and the other when the left hand side is not
a variable.  This will be used to illustrate that the formalism is able to
accommodate languages where equations need not be directed.  The second type of equation is the discrete assignment.  ``$x$ is
reset to e''.  Discrete assignments are essential for modeling hybrid
systems, where instantaneous changes of a value (resets) can occur in
juxtaposition to continuous dynamics.    The third type of equations is a conditional equation, which allows us
to express the choice between which of two equations holds depending on
the boolean condition given as an expression.  The fourth kind of equations is a universal quantification, and it
provides a concise way of describing the dynamics of a system that has
a family of state variables. The variable introduced by this construct
may only be an unprimed name.  The last construct is a set of
equations $\sets{s_j}$.

\subsection{Type System}
A \CoreDEL expression $e$ has type $\type$ under environment $\tenv$
when the judgement $\tenv \vdash e:\type$ is derivable according to
the rules presented
 in Fig.~\ref{fig:Lsyntaxtype}.  The rules for
natural, real, and boolean constants are straightforward.  The rule
for unprimed names is simple environment lookup.  The rule for primed
variables, however, requires that both primed and unprimed variables
have type real. The rule for vector construction is also
straightforward. Vector indexing is a bit more interesting, as it
treats the case when the index is a literal as a special case,
allowing elements to have different types.  This makes it possible to use
vectors both for tuples and for (homogeneous) vectors.  Function
applications assume that we have a function $n_f$ that determines the
arity of the function $f$, and a function $\tau_{f,j}$ that determines
the type of the $j$th argument to the function.  Partial derivatives
have straightforward rules.  The rules for equations are straightforward. Finally, environment extension of environment $\Gamma_1$ with the binding
$x:\type$, written $\Gamma_1,x:\type$ is an environment $\Gamma_2$
defined as follows:
\[
\Gamma_2(y) = 
\begin{cases}
	\type &\mbox{if} \ y = x,\\
  	\Gamma_1(y) & otherwise	.		
\end{cases}
\]

\subsection{Example:  A Lagrangian Model}

\begin{figure}[h]
 \caption{Compiling Pendulum/Mass Example}
    \centering
    \subfloat[Latex-style Acumen-17 Source for Pendulum/Mass Example]{
     \begin{tabular}{| >{$}l<{$}|}
        \hline
        q = (x,\theta), a = 1, m = 2, M = 5, g = 9.8,  \\ \\ 
         k = 2, I = \frac{4}{3}ma^2, L=T - V, \\ \\ 
         T = \frac{1}{2}(M + m)\dot{x}^2 + ma\dot{x}\dot{\theta}\cos(\theta) + \frac{2}{3}ma^2\dot{\theta}^2 \\ \\ 
         V = \frac{1}{2}kx^2 + mga(1-\cos(\theta)), \ \  \\ \\
         \forall i\in\{1...|q|)\}, \frac{d}{dt} \left(\frac{\partial L}{\partial {\dot{q_i}}}\right)
 - \frac{\partial L}{\partial{q_i} }=0  \\\hline  
    \end{tabular}}
    \qquad
     \subfloat[After Binding-Time Analysis (BTA)]{
     \begin{tabular}{| >{$}l<{$}|}
        \hline
        \rowcolor[gray]{0.9}
         q = (x, \theta), a = \colorbox{white!20}{1}, ...  I = \colorbox{white!20}{$\frac{4}{3}ma^2$}\\ 
        \rowcolor[gray]{0.9}  \\ 
         \rowcolor[gray]{0.9}  
         \colorbox{white!20}{$\forall i\in\{1...|q|)\}$}, \frac{d}{dt} \left(\frac{\partial L}{\partial {\colorbox{white!20}{$\dot{q_i}$}}}\right)
 - \frac{\partial L}{\partial \colorbox{white!20}{\makebox(1,1){$q_i$}} }=0  \\\hline      
    \end{tabular}} 
   \qquad
    \subfloat[After Specialization (Implicit ODEs)]{
      \begin{tabular}{| >{$}l<{$}|}
         \hline
         q = (x,\theta), a = 1, ...  I = \frac{8}{3} \\ \\ 
         2\cos(\theta)\ddot{\theta} - 2\sin(\theta)\dot{\theta}^2 + 7\ddot{x} + 2x = 0 \\ \\
         \frac{98}{5}\sin(\theta) + 2\cos(\theta)\ddot{x} + \frac{8}{3}\ddot{\theta} = 0 \\ \\ 
          \hline   
     \end{tabular} }
     \qquad 
      \subfloat[After Symbolic Gaussian Elimination (Explicit ODEs)]{
      \begin{tabular}{| >{$}l<{$}|}
         \hline
         A = \sin(\theta), B = \cos(\theta), 
         \\ \\
         \ddot{x} = \frac{2(A\dot{\theta}^2 - x) - B\ddot{\theta}}{7}
          \\\\
         \ddot{\theta} = \frac{\frac{-686}{5}A - 4B(A\dot{\theta}^2-x)}{\frac{56}{3} - 4B^2}
         \\ 
          \hline   
     \end{tabular} }
\label{fig:compile}
\end{figure}

For a variety of technical reasons, researchers working on novel robotic systems tend to make extensive use of the Lagrangian method. It is especially useful in the case when the system being described has more than one state variable.  Then modeling using Lagrange employs families of equations, which are written as one equation but really represent a collection of different equations derived by instantiating certain indices.  Figure~\ref{fig:compile}(a) provides the Acumen-17 model of a second order nonlinear system shown in Figure~\ref{fig:springmass}.  It consists of a pendulum hanging from a mass, which is attached through a  spring to a wall.  

\begin{wrapfigure}{h}{0.3\textwidth}
\includegraphics[width=0.3\textwidth]{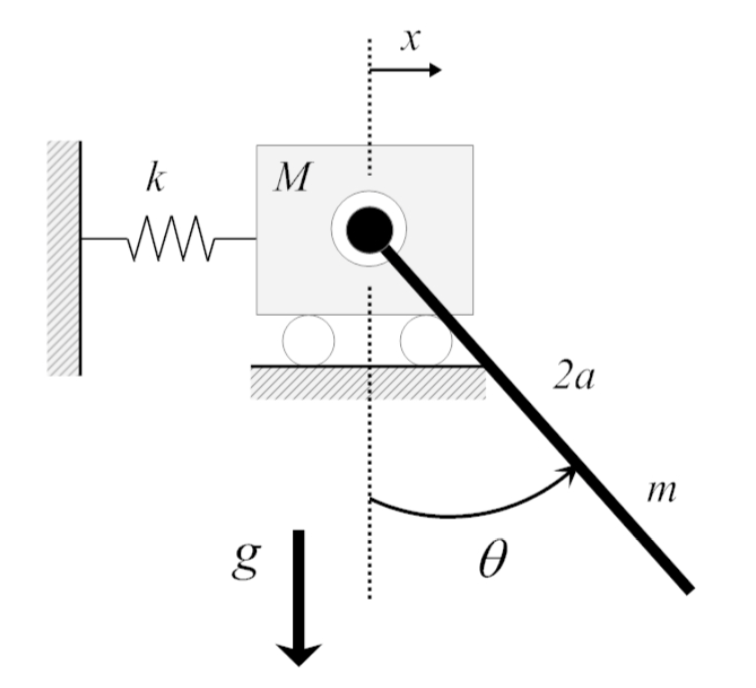}
\caption{\label{fig:springmass}A Pendulum Spring-Mass system.}
\end{wrapfigure}
As the system has two degrees of freedom, $x$ and $\theta$, the example introduces a vector of state variables $q$.  The Euler-Lagrange equation that appears at the end of the example is expressed by the family of equations.  In Figure~\ref{fig:compile}(a),  the $\forall$ quantifier is used to introduce the index variable for a family of equations.  In the  ASC\RNum{2}-based syntax, the keyword
\verb|foreach| represent this quantifier.  The intent is to express as concisely and as close to what would typically appear in a mechanics textbook:
\[
\frac{d}{dt} \left(\frac{\partial L}{\partial {\dot{q_i}}}\right) - \frac{\partial L}{\partial{q_i} }=0 \ \ \ \text{where} \ \ \
\begin{aligned} \begin{split}
 q= (x,\theta), \\
 i \in \{1,2\}. \end{split} \end{aligned}
\]

This notation generally has a syntactic interpretation, that is, the $name$ contained in the $i$th element of the vector is looked up.
In other words, this family of equations literally represents the following two equations:
\[
\frac{d}{dt} \left(\frac{\partial L}{\partial {\dot{x}}}\right) - \frac{\partial L}{\partial{x} }=0 \ \ \  \text{and} \ \ \ 
\frac{d}{dt} \left(\frac{\partial L}{\partial {\dot{\theta}}}\right) - \frac{\partial L}{\partial{\theta} }=0
\]

The offline partial evaluation strategy enables us to  support family of equations and partial differentiation by utilizing the two 
most important components, namely the binding-time analysis (BTA) and specialization.  A successful BTA annotates the model
with instructions for performing certain part of the computation early and other part for later processing.  The annotated model 
for the pendulum/spring example is illustrated in Fig.~\ref{fig:compile}(b).  In this illustration, computations that remain
for further processing are shaded grey, whereas computations can be performed immediately in the next specialization phase appear in
a white background.  The value of $a$ is marked known or static as it is a value, and the BTA also annotates variable $I$ 
known for that both $m$ and $a$ are known variables.  A more interesting case is the indexing operator $q(i)$ .
Although the state variable vector $q$ being marked unknown, in fact we need to solve for the values of state variables $x$ and $\theta$
in the simulation, we can still perform this operation statically for the reason that the size of $q$ and the index variable $i$ are known.

The step which performs the work that a BTA schedules is called specialization.  The result of specializing our running example is presented in Fig.~\ref{fig:compile}(c).  Computing the value of $I$ is simple rational arithmetic.  The instantiation of a  family of equations is essentially a type of iteration, which also replaces $q_i$ by $x$ and $\theta$ by vector lookup.  In the same time, symbolic time and partial
differentiation are performed using the chain rule.  Solving multiple implicit ODEs to explicit form equations is achieved using an analog of symbolic Gaussian elimination. For our running example, the result of this step is presented in Fig.~\ref{fig:compile}(d).  Abstract interpretation with interval analysis is used to ensure the pivot expression is non zero.
To control the  system, for example, stabilizing the position and the angular displacement, one can add two PD controllers. The modification to the original model in Fig.~\ref{fig:compile}(a) are as follows: 
\begin{table}[h!]
\centering
 \begin{tabular}[h!]{| >{$}l<{$}|}
 \hline
       
ux = 100*(2 - x) + 30*(0 - \dot{x}), ut = 100*(\pi - \theta) + 40*(0 - \dot{\theta}), \\ \\  
u = (ut, ux), \quad
 \forall i\in\{1...|q|)\}, \frac{d}{dt} \left(\frac{\partial L}{\partial {\dot{q_i}}}\right)
 - \frac{\partial L}{\partial{q_i} }=u(i) \\ 
          \hline   
     \end{tabular} 
 \end{table}

The Acumen implementation supports an enclosure simulation semantics that produces rigorous over-approximations (guaranteed upper and lower bounds) for all simulations \cite{adam}. Previously, this implementation only supported a formalism that worked with hybrid ODEs. With the work we presented here, this implementation can now handle models such as the pendulum spring mass model presented above.  The plot of controlled system are as follows:
 
\begin{figure}[h]
    \centering
    \subfloat[$\theta$]{
     \includegraphics[width=0.24\textwidth]{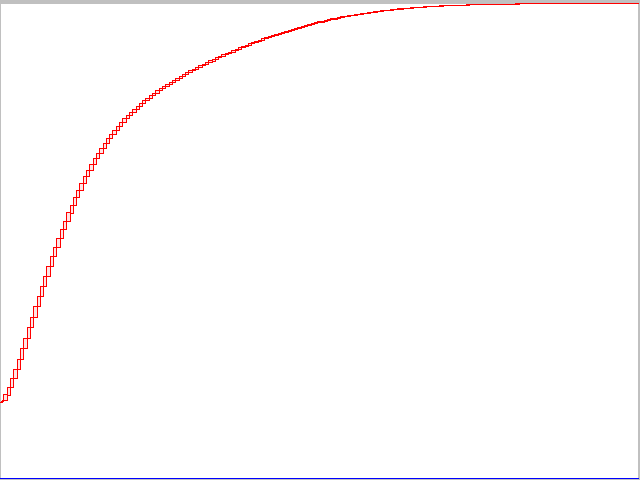}
   }
   \subfloat[$\ddot{\theta}$]{
     \includegraphics[width=0.24\textwidth]{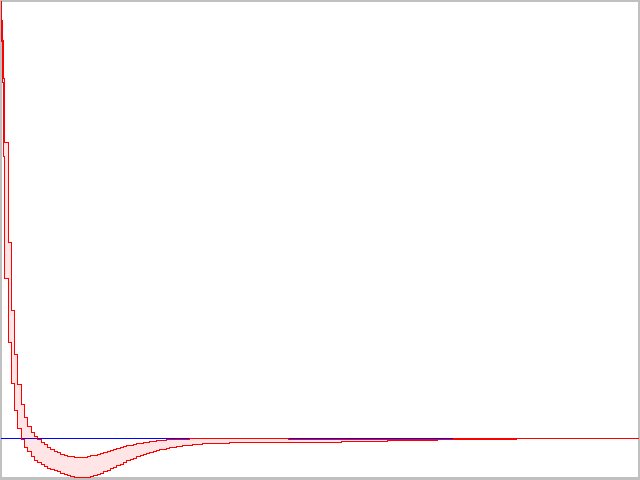}
   }
\subfloat[$x$]{
     \includegraphics[width=0.24\textwidth]{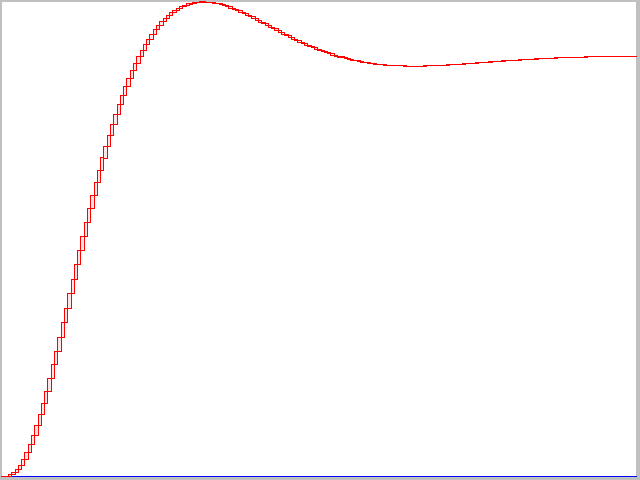}
   }
\subfloat[$\ddot{x}$]{
     \includegraphics[width=0.24\textwidth]{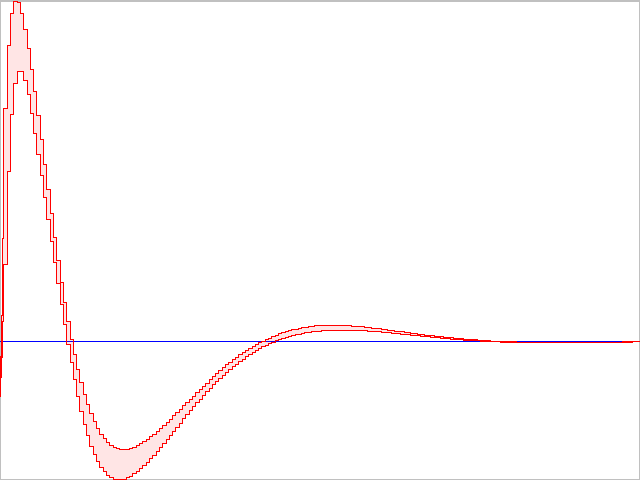}
   }
\end{figure}
\begin{figure}
    \caption{Two Case Studies}
    \centering
    \subfloat[Cross section of A Cam and Follower]{  
  \includegraphics[width=0.3\columnwidth]{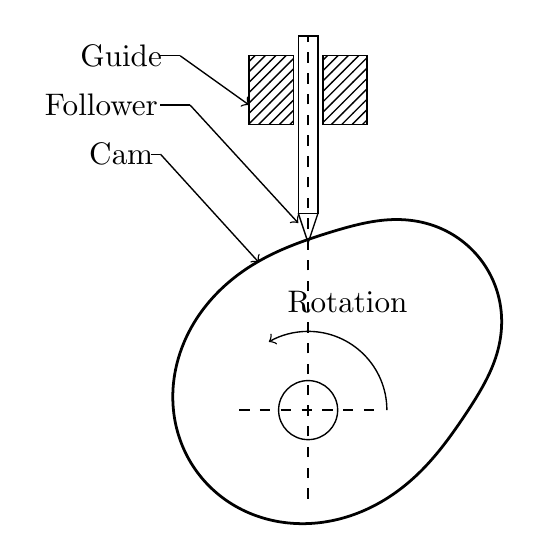}\label{figure:cam}
    }
   \qquad
     \subfloat[A Compass Gait Biped \cite{goswamicompass}]{  
   \includegraphics[width=0.4\columnwidth]{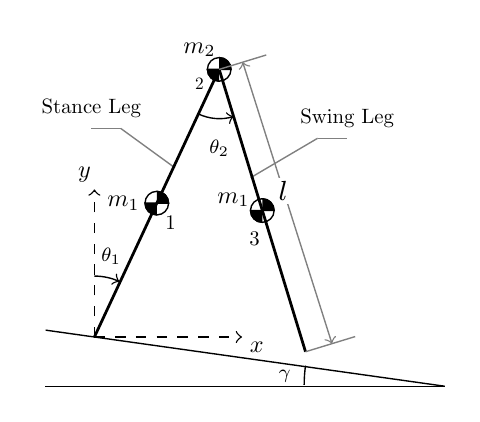}
    
    }
    \label{casestudy}
\end{figure}

\captionsetup[subfigure]{subrefformat=simple,labelformat=simple,listofformat=subsimple}
\renewcommand\thesubfigure{(\alph{subfigure})}
\subsection{A Cam and Follower Example}
We further demonstrate the expressiveness of the proposed language using the following two case studies. Transforming rotational motion into any other motions is often conveniently accomplished by means of a \emph{cam mechanism}.  A \emph{cam} is defined as a machine element having a curved outline, which by its rotation motion,
gives a predetermined motion to another element, which is often called \emph{follower}.  Fig~.\ref{casestudy}(a) shows such a \emph{cam mechanism}, the curved outline of cam $r$ is a function of rotational angle $\theta$, defined as below:
\[r = (1.5-\frac{cos\theta}{2}) * (1 + \frac{cos(2\theta)}{5} ) \]

In the study of various aspects of the follower motion, the velocity and acceleration of the follower are needed.  To get 
the correct form, the modeler usually has to manually derive the partial derivatives.  Fig.~\ref{fig:armcode} in the Appendix  shows the mathematical model and the corresponding \CoreDEL program.   Clearly, supporting partial derivatives in the language greatly simplifies the modeling task, and can save the modeler much tedious and error-prone work.

\subsection{A Compass Gait Biped Example} 
The Compass gait biped model \cite{cornellBiped,Aaron2D} is a two dimensional unactuated rigid body system placed on a downward surface inclined at a fixed angle $\gamma$ from the horizontal plane. A diagram of the model is shown in Fig~.\ref{casestudy}(c) with its physical parameters.  The configuration of this two-link
mechanism can be described by the generalized coordinates $q = [\theta_1, \theta_2]$, where $\theta_1$ is the angle from the vertical line to the \emph{stance leg} and $\theta_2$ is the angle between two legs.
 It is a hybrid model featuring two phases. At the start of each step, the system is governed by its
 continuous dynamics until the \emph{swing leg} hits the ground.  The discrete event can be modeled as an inelastic collision conserving angular momentum.  The stance and swing legs switch instantaneously  during the collision and go into the next step after.
\subsubsection{Continuous Dynamics and Discrete Event}
The continuous phase of this system can be modeled using the same Lagrange method shown earlier.
Let point $(x_i, y_i)$ denote the position of centralized masses shown in Fig~.\ref{casestudy}(b), form which its easy to define the kinetic and potential energy of the system.
Applying the same Lagrange equation shown in Equation 1 with $q = (\theta_1, \theta_2)$, we have the dynamic equations of the system during the swing phase.  The perpendicular distance from the walking surface to the tip of the \emph{swing leg} is given by
\[guard = lsin\gamma(sin\theta_1 + sin(\theta_2 - \theta_1) )  \]

Where $\gamma$ is the slope of the ground.  Impact occurs when the tip of the swing leg hits
the walking surface in a downward direction, which can be describe as follows:
$guard \leq 0 \wedge \dot{guard} < 0$.  Using conservation of angular momentum \cite{goswamicompass}, the explicit solution of post impact velocities can be determined.  Fig.~\ref{fig:armcode} in the Appendix  shows the mathematical model and the full \CoreDEL model.  This example shows the proposed formalism can support a direct mapping from mathematical model to simulation code for a hybrid system model with complex dynamics.

\begin{figure*}[h]
\footnotesize
\centering
\begin{tabular}{ l >{$}c<{$} >{$}l<{$}  >{$}l<{$}  }
%\fbox{\text{Syntax}} & & & \\
&&&\\
%value level zero &\avz& -> & i \ | \ q \ | \ d\  | \sequence{v_j} \\
Binding Time & b    & ::=& \stat \ | \ \dyn \\
$\text{Expression}$&e^b&::=& \const^{b}\  |\  x^{b} \ | \ (\sequence{e_j^{^{b_j}}})^b \ | \
                                           e_1^b(e_2^b)^b \ | \ \f(e^b)^b \ | \ 
                                           (\der{e^{b}})^{b} \ | \ (\pder{e_1^{b_1}}{e_2^{b_2}})^{b}\\
$\text{Equation}$&s^b &::=& (x^b = e^b)^b\ | \ ({x^{+}}^b = e^b)^b\ |\  (e_1^b = e_2^b)^b ,\ e_1 \nequiv x \\
                            & &&              | \ (\tb{if} \ e^b\ \tb{then} \ s_1^b \ \tb{else} \ s_2^b)^b \ | \ 
                                                 (\forall n^b \in e^b. \ s^b)^b  \ | \ (\sets{s_j^{b_j}})^b  \\
                                                 
Binding Time type &\type^\bt & ::= &  \tb{nat}^b \ | \ \tb{bool}^b \ | \ \tb{real}^b\ \ | \ \vbtype{b}{\type_j^{b_j}}  \\

Binding Time Environment &\btenv& = &  \sets{x_j : {\type_j}^{b_j}} \ \text{and}\ x_i = x_j \implies i = j \\
%&&&\\
\end{tabular}
\begin{displaymath}
\begin{array}{l}
\begin{array}{c}
%\multicolumn{1}{l}{\fbox{$\btenv \vdash e^b:\type^{b}$} }\\
 \infer{\btenv \vdash i^b: \tb{nat}^b} {} \quad \infer{\btenv \vdash q^b: \tb{real}^b} {} \quad \infer{\btenv \vdash t^b: \tb{bool}^b} {} \quad
 \infer{\btenv \vdash n^b : \btt{b}} {n^b:\btt{b} \in \btenv } \quad 
 \infer{\btenv \vdash x'^{b} : \tb{real}^{b}} {x'^b:\tb{real}^{b} \in \btenv & x^\dyn:\tb{real}^{\dyn} \in \btenv } \quad

\\
\infer{\btenv \vdash (\sequence{e_j^{b_j}})^b : \vbtype{b}{\type_j^{b_j}}} {\btenv \vdash e_j^{b_j} : \type_j^{b_j} \quad b = \sqcup b_j} \quad
\infer{\btenv \vdash e_1^{b}(i^\stat)^{b_i}:\type_{i}^{b_i}} 
	  {\btenv \vdash e_1^{b} :\vbtype{b}{\type_j^{b_j}}\ \\ \btenv \vdash i^{\stat}:\tb{nat}^{\stat} \ \  i < m}\quad
\infer{\btenv \vdash e_1^{b_1}(e_2^{b_2})^b:\btt{b}} 
   	  {\btenv \vdash e_1^{b_1} :\vbtype{b_1}{\type^{b_1}}\ \\ \multicolumn{1}{c}{\btenv \vdash e_2^{b_2}:\tb{nat}^{b_2}} \\ \multicolumn{1}{c}{b = b_1 \sqcup b_2}} \\
\infer{\btenv \vdash \f(e^b)^b : \type^b} 
	  {\f:\vtypef{\type_{f,j}} ->\type \\ \btenv \vdash e^b:\vbtypef{b}{\type_{f,j}^{b_j}}} \quad 
\infer{\btenv \vdash (\der(e^b))^b : \tb{real}^{b}} {\btenv \vdash e^b : \tb{real}^{b}} \quad
\infer{\btenv \vdash (\pder{e_1^{b_1}}{e_2^{b_2}})^b : \tb{real}^{b}} {\btenv \vdash e_1 : \tb{real}^{b_1} \quad \btenv \vdash e_2 : \tb{real}^{b_2} \\ 
\multicolumn{1}{c}{b = b_1 \sqcup b_2}}\\
\end{array}\\ \\
%\multicolumn{1}{l}{ \fbox{$\btenv \vdash s^b $}}\\
\begin{array}{c}
\infer{\btenv \vdash (x^b = e^b)^b} { x^b: \btt{b} \in \btenv & \btenv \vdash e^b:\btt{b}}\quad
\infer{\btenv \vdash ({x^{b_1}} {}^{+}= e^{b_2})^{b_2}} { \btenv \vdash : x  : \btt{b_1} & \btenv \vdash e:\btt{b_2}}\quad
\infer{\btenv \vdash (e_1^{b_1} = e_2^{b_2}) ^{b}} { \btenv \vdash e_1^{b_1} : \btt{b_1} & \btenv \vdash e_2^{b_2}:\btt{b_2} & b = b_1 \sqcup b_2} \\
\infer{\btenv \vdash (\tb{if} \ e^{b}\ \tb{then} \ s_1^{b_1} \ \tb{else} \ s_2^{b_2})^{b_3}}
	  {\btenv \vdash e : \tb{bool}^b \quad \btenv \vdash:s_1^{b_1} \quad \btenv \vdash s_2 ^{b_2} \quad b_3 = b \sqcup b_1 \sqcup b_2} \quad
\infer{\btenv \vdash (\forall n \in e^{b_1} \ . \ s^{b_2})^{b_2}}
	  { \btenv \vdash e:\vbtype{b_1}{\type^{b_j}} \\  \btenv,n:\type^{b_1}  \vdash s^{b_2} }\quad
\infer{\btenv \vdash (\sets{s_j^{b_j}})^{b}} { \btenv \vdash s_j ^{b_j} \quad b = \sqcup b_j}  \qquad
\end{array}
\end{array}
\end{displaymath}
\caption{Binding Time Analysis for \CoreDEL}
\label{fig:LBsyntaxtype}
\end{figure*}
\section{BTA and Specialization for \CoreDEL}\label{section:BTASpecialization}
This section presents a declarative specification of the binding-time
analysis (BTA) and specialization process for \CoreDEL, and proves the correctness (type-safety) of the BTA with respect to the specialization process.

\subsection{BTA}
BTA is the analysis performed in an offline partial evaluation system to determine, given some early or
``static'' inputs to a program, which of the program's computation can be done at an early stage \cite{Jones85}. Fig.~\ref{fig:LBsyntaxtype} gives a declarative specification of the BTA.  There are two binding times $\stat$ and $\dyn$, representing
``static'' and ``dynamic'' computations, respectively.  Static is for
compile-time computations that are done before the simulation starts,
and dynamic is for computations done during the simulation proper.
Expressions, equations, types, and type environments are all annotated
with binding times. 

The changes to the derivation rules are largely straightforward.
Essentially, binding times are propagated with types.  In addition,
when multiple subexpressions occur, their binding times are combined
using the least upper bound operator $\sqcup$ which returns static only
if all arguments are static, otherwise returns dynamic.  However, for vector indexing, when the index expression is static and the subexpressions dynamic, we will still perform the look up operation.  Finally, the rule 
for primed variable requires the unprimed variable with the same name to be dynamic. And in
the rule for vector indexing with a literal, where the literal is
annotated as static, the binding time of the expression is the same as 
the corresponding entry.
\subsection{Specialization}\label{section:specialization}
 Fig.~\ref{fig:semantics} in Appendix presents the big-step semantics for the specialization process. 
 Values include constant with static annotation, dynamic expression and vector of values.  
 Normal form equations are straightforward, with the absence of universal quantification equation.   
The first auxiliary function is used to compute static function application.  The last two are for eliminating total and partial derivatives using the chain rule.   Function $FV$ returns free variables in an 
expression and function $LV$ extracts the left hand side variable of a directed equation.

Relation $\ereduce$  essentially specializes all subexpressions to values then combines them according to their binding times.  Function application with static binding time returns the evaluation result of the corresponding function. Vector indexing with static index performs look up operation statically, even if the vector itself has dynamic binding time.
Total and partial derivatives get eliminated statically using different rules depending on the what their subexpressions specialized to.  The rules for relation $\sreduce$ are similar.  However, they all require that the equation to be specialized does not contain free variables that are defined in the equations following it.  For directed equations, in addition to specializing the right hand side expression, the rule also substitutes the result into the rest of the equations.
Both relations can also generate \err \ terms, which will be propagated to top level for error reporting. For example, the static index may be specialized to a natural number that is bigger than the size of the vector.  However, 
one type of error we do not catch is the case of partial derivative $\pder{e_1^{b_1}}{e_2^{b_2}}$, when $e_2^{b_2}$ can not be specialized to a variable $x^{b_2}$. It is analogous to the traditional division by zero error. 

\subsection{Type Safety}
\theoremstyle{definition}
\newtheorem{defn}{Definition}
\newtheorem{theorem}{Theorem}
\newtheorem{lemma}{Lemma}
 \begin{defn}\label{def:sim} The erasure relation $|\cdot|$ for $e^b$, $s^b$ and $\btenv$ is defined as follows:
\[ | e^b| = e \quad |s^b | = s \quad  |\btenv|(x) = \type \ \text{if}\  \tenv(x) = \type \]
 \end{defn}
%\newdef{lemma}{Lemma}
\begin{lemma}[Erasure preserves typablity] $\forall \btenv, e^b, s^b$
\[   \btenv \vdash e^b:\type^b  => |\btenv| \vdash |e|:\type\]
\[  \btenv \vdash s^b => |\btenv| \vdash |s|\]
 \end{lemma}
\begin{proof}
By induction on the derivation of $\btenv \vdash e^b:\type^b$ and $\btenv \vdash s^b$, respectively.
\end{proof}

\begin{lemma}{Substitution type preservation}  $\forall \btenv,$ \\ $ x, e, s,\type$.
\[ 
\begin{array}{l}
 \btenv \vdash v^{b_1}:\type^{b_1}  \ \wedge \ \btenv, x:\type^{b_1} \vdash e:\type^{b_2}
	 \\  \qquad \implies \btenv \vdash e[x :=v^{b_1}]^{b_2}:\type^{b_2} \\
  \btenv \vdash v^{b_1}:\type^{b_1}  \ \wedge \ \btenv, x:\type^{b_1} \vdash s^{b_2}
	\\ \qquad \implies \btenv \vdash s[x := v^{b_1}]^{b_2} 
\end{array}
\]
\end{lemma}
\begin{proof}
By induction on the derivation of $\btenv \vdash e^b:\type^b$ and $\btenv \vdash s^b$ respectively.
\end{proof}

\begin{lemma}[Type Preservation] $ \forall \tenv,\btenv,e,s,\type.$\label{lemma:typepreserve}
\[\btenv \vdash e^b:\btt{b} \wedge e^b \ereduce \vb \implies
    \btenv \vdash \vb:\type\]
\[\btenv \vdash s^{b} \wedge s^{b} \sreduce \wb \implies
    \btenv \vdash \wb\]
\end{lemma}

\begin{lemma}{Static value}  $\forall \btenv, e.$ \label{lemma:staticvalue}
    \[\btenv \vdash v^\stat:\tb{nat}^\stat \implies |v^{\stat}| = i  \]
    \[\btenv \vdash v^\stat:\tb{bool}^\stat \implies |v^{\stat}| = t  \]
    \[\btenv \vdash v^\stat:\tb{real}^\stat \implies |v^{\stat}| = q  \]
\end{lemma}
\begin{proof}
When the binding time of a value $v^b$ is static, by the  definition, $v$ can only be a constant or a vector of constants. And by typing rules in Fig.~\ref{fig:LBsyntaxtype}, we can prove the lemma above.
\end{proof}

\begin{theorem}{Type safety of specialization} Let $\btenvd$ denote $\sets{x_j:\btt{\dyn}}$ and $\forall \btenvd, e, s, b,\vb,\done$ 
\[\btenvd \vdash e:\type^b \wedge e^{b} \ereduce r \implies r \neq \err \wedge \btenvd \vdash r:\type^b  \]
\[\btenvd \vdash s^b \wedge s^{b} \sreduce r \implies r \neq \err \wedge \btenvd \vdash r \]
\end{theorem}
Due to the page limits, the proof can be found in Section 1 of the supplementary document.

\section{Algorithmic Specification}\label{section:Implementation}
This section presents an algorithmic specification of the BTA.  First we introduce the constraint type and the constraint 
generation function. We proceed by providing a normalization method that guarantees to find the unique
\emph{minimal solution}, if one exists. At last,  we show that specification is faithful to the declarative BTA.
 
\subsection{Generating Binding Time Constraints}  
Fig.~\ref{fig:constraint-gen} defines Label $l$, which can be \tb{root} or a label indexed by a natural number.  For example, let the label of equation $x = 1$ be $l$, and the label for $x$ and $e$ be $l_1$ and $l_2$ respectively.  A Binding Time Expression $B$ can be static, dynamic or a label. A constraint $c$ is a partial order $\nlt$ between two binding time expressions.  Constraint is satisfied in three cases $\stat \nlt \dyn, \stat \nlt \stat \ \text{and}\  \dyn \nlt \dyn $.

\subsection{Control Flow and Control Scope}
Because of conditional equations, a variable may have a different binding time depending on where it appears.  For example:
\begin{verbatim}
x = 1, if t < 5 then y = x  else y' = x
\end{verbatim}
The value of \textit{x} is statically known for both branches but  value of \textit{y} is only static in the first branch. 
To handle the control scope issue, we build an auxiliary \emph{global environment}  while labeling the program. 
\begin{defn}\label{def:scopeenvironment}
A total map $\scopeenv: \text{Variable} \to \labels \ \cup \ {\dyn}$ is called \emph{local environment}.\qed
\end{defn}
\begin{defn}\label{def:labelling}
 A map $\rho: l \to \scopeenv$ is called \emph{global environment}.  To look up the defining label of a variable $x$
 inside scope $\scope$, we first find the corresponding \emph{local environment} $\scopeenv$, then apply variable
 x in it.  That is to say $\rho(\scope)(x)$, we  abbreviate it to $\rho(\scope,x)$ in the rest of this section.\qed
 \end{defn}
 Let $\labelE : s \to \rho$ be a construct function that takes a  equation $s$ and returns a
 \emph{global environment} defined as follows, assuming the label of each equation is $l$, the scope is $\scope$ and
 starting global environment be $\rho$:
 \[
 	\begin{array}{l}
 		\labelE(\scope, \sets{s_j}, \rho) =  \labelE(\scope, \{s_j\}^{j \in 2...m}, \labelE(\scope,s_1,\rho))\\
		\\
		\labelE(\scope, \setof{x = e}_l, \rho) = \rho \uplus (\scope \to \rho(\scope) \uplus (x \to l_1))\\
				 \\
	 	\labelE(\scope, \{\tb{if}\ e \ \tb{then}\ s_1 \ \tb{else}\ s_2\}_l,\rho) =  \\
		\qquad  \rho \uplus \labelE(l_2, s_1,\rho)
														    \uplus \labelE(l_3, s_2,\rho)\\									   
 	\end{array}\]

When constructing a \emph{global environment}, the starting scope label $l$ is \tb{root}, and changed to the label of branches when inductively constructing inside a conditional equation. The example above has the following \emph{global environment}:                                                                    
\[
\begin{array}{l}
				 \{\tb{root} \to \{x \to {\tb{root}_1}_1\}, {\tb{root}_2}_2 \to \{x \to {\tb{root}_1}_1, y -> {{{{\tb{root}_2}_2}}_1}_1,\}\\
				\ \  {\tb{root}_2}_3 \to \{ x \to {\tb{root}_1}_1, y' \to {{{{\tb{root}_2}_3}}_1}_1 \} \}
\end{array}\]

\begin{defn}\label{def:constraint-sets}
	$C = \sets{c_j}$  is called a \emph{constraint set}.
	The \emph{labels} of $C$ are defined as
	$$\labelSet{C} = \setbar{l}{\exists \, B \mbox{ st } l \nlt B \in C \mbox{ or } B \nlt l \in C}.$$

	\noindent
    $C$ is in \emph{normal form} if for all $c\in C$, one of the following is true: 
	\[c \equiv \stat \nlt l,\   c \equiv l \nlt \dyn,\ c \equiv l \nlt \tilde{l}\]
	and is of \emph{error form} if
	\(\dyn \nlt \stat\ \in C\).
	The sets of $C$ in normal form and of error are denoted by $\CNF$ and $\CError$, respectively.
	\qed
 \end{defn}

\begin{figure*}[ht]
$
\small
\centering
\begin{array}{l}
\multicolumn{1}{l}{\fbox{\text{Syntax}}}\\
%\begin{center}
\begin{tabular}{   >{$}r <{$} >{$} c<{$} >{$}l<{$}   }
%\begin{tabu}{|rcl|}

\labels \quad l & := & \tb{root} \ |\  {l_i} \ \text{where} \  i \in \mathbb{N} \\
\text{Binding Time Expression} \quad B & := & b \ | \ l \\
\BTC  \quad c&:=& B \nlt B \\
\text{Satisfied Constraint}  \quad \vdash c&:=& \vdash \stat \nlt \dyn \ | \  \vdash \stat \nlt \stat \ | 
								       \vdash \dyn \nlt \dyn   \\							   
\end{tabular}
\\ \\
\multicolumn{1}{l}{\fbox{\text{Constraint Generation Function}}}\\
\\
%\end{center}
%\begin{center}
\def\arraystretch{1.5}
\begin{array}{l}
\begin{tabular}{ >{$} l <{$} | >{$} l <{$} |}

\text{Input}\quad	e 					& \text{Output}\quad |[e|]_{\scope, l, \rho}			 \\
	\hline
	k					& \setof{l \nlt \stat} 											 \\
	n					& \setof{\rho(\scope, n) \nlt l, \ l \nlt \rho(\scope, n)} 		 \\
	x'					& |[ x|]_{\scope, l_1, \rho} \ 
						  \cup \setof{\rho(\scope, x') \nlt l, \dyn \nlt l_1 }			 \\
	\sequence{e_j}		& \bigcup_\indexing |[e_{j}|]_{\scope, l_j, \rho} \ 
						  \bigcup_\indexing \setof{l_j \nlt l}							 \\
	e_1(e_2)			& |[e_1|]_{\scope, l_1, \rho} \ 
						  \cup |[e_2|]_{\scope, l_2, \rho} \ \\ & 
						  \cup \setof{l_1 \nlt l, \ l_2 \nlt l}							 \\
	f(e)					&  |[ e|]_{\scope, l_1, \rho} \  
						  \cup \setof{l_1 \nlt l}	 						 \\
	\der{e}				& |[ e|]_{\scope, l_1, \rho} \  
						  \cup \setof{l_1 \nlt l}		 								 \\
	\pder{e_1}{e_2}		& |[e_1|]_{\scope, l_1, \rho} \ 
						  \cup |[e_2|]_{\scope, l_2, \rho} \ \\ &  
						  \cup \setof{l_1 \nlt l, \ l_2 \nlt l} 						 \\

	\end{tabular} \quad
\begin{tabular}{| >{$} l <{$} | >{$} l <{$} }
\text{Input}\quad	s 					  &  \text{Output}\quad |[s|]_{\scope, l, \rho} \\
	\hline                
	x = e 				  & |[x|]_{\scope, l_1, \rho} \ 
						    \cup |[e|]_{\scope, l_2, \rho} \\ &  
						    \cup \setof{l_1 \nlt l, \ l_2 \nlt l, \ l_2 \nlt l_1}  
																 \\
 	x^{+} = e 				  & |[e|]_{\scope, l_2, \rho} \ 
						    \cup \setof{l_2 \nlt l}										 \\
	e_1 = e_2, e_1 \neq x & |[e_1|]_{\scope, l_1, \rho} \ 
							\cup |[e_2|]_{\scope, l_2, \rho} \ \\& 
							\cup \setof{l_1 \nlt l, \ l_2 \nlt l}						 \\
	\tb{if} \ e \ \tb{then} \ s_1 \ 
						  & |[e|]_{\scope, l_1, \rho} \ 
						    \cup |[s_1|]_{l_2, l_2, \rho} \ 
						    \cup |[s_2|]_{l_3, l_3, \rho} \ \\ \tb{else} \ s_2 & 
							\cup \setof{l_1 \nlt l, \ l_2 \nlt l, \ l_3 \nlt l}			 \\
	\forall n \in e \ . \ s 
						  & |[n|]_{\scope, l_1, \rho} \ 
						   
						    \cup |[s|]_{\scope, l_3, \rho[(\scope, n) \mapsto l_1]} \ \\ &
							 \cup |[e|]_{\scope, l_2, \rho} \ \cup \setof{l_2 \nlt l_1} \ 
						    \cup \setof{l_3 \nlt l}										 \\
	\sets{s_j}			  & \bigcup_\indexing |[s_j|]_{\scope, l_j, \rho} \ 
						    \bigcup_\indexing \{l_j \nlt l\}							 \\

\end{tabular}
\end{array}
%\end{center}
\end{array}$
\centering\caption{Constraint Generation}\label{fig:constraint-gen}
\end{figure*}

Fig.~ \ref{fig:constraint-gen} defines function $|[\cdot|]$ that takes $e$ in scope $\scope$ and \emph{global environment} $\rho$ returns a constraint set.   
It generates a constraint $l \nlt \stat$ for constant $k$.  For an unprimed variable $n$, it generates a  constraint between the occurrent label $l$ and the 
definition label in the \emph{global environment} $\rho(\scope,n)$.  For variable with
primes $x'$,  additional constraints between $\dyn$ and labels of all its lower derivatives in the \emph{global environment} are added. 
 For $\vectoromit{e_j}$, $e_1(e_2)$,$ f(\vectoromit{e_j})$,$\der(e)$ and $\pder{e_1}{e_2}$ , we generate constraints between label of  $e$ and all its subexpressions $e_i$, that is $\bigcup_{i}  \{l_i \nlt l\}$.  Then inductively apply $|[\cdot|]$ to all the subexpressions.\\
Function $|[\cdot|]$ for equation is defined in a similar way.
In the case of  directed equation $n = e$ and $x' = e$, the binding time of  corresponding definition label depends on right hand side expression $e$, captured by $l_1 \nlt \rho(\scope,n)$ and $l_1 \nlt \rho(\scope,x')$ respectively.  For $x ^{+}= e$, $e_1 = e_2$ and 
$\setsomit{s_j}$, the binding time depends on subexpressions and sub-constraints. 
In the case for $\tb{if}\ e \ \tb{then} \ s_1 \ \tb{else}\ s_2 $, it changes the scope label from $\scope$ to $l_2$ and $l_3$ when inductively apply $|[\cdot|]$ to $s_1$ and $s_2$. In the case of $\forall n \in e, \ s$, it adds a new mapping from binding variable into the \emph{global environment} $\rho$ when inductively generate constraint from $s$.
\begin{defn}\label{def:substitution}
	A map $\sigma\colon \labels \to \BT$ is called a \emph{substitution}, and $\sigma$
	is the identity function on labels not in domain of $\sigma$.
	As $\labels$ is finite, $\domain{\sigma}=\{l_1,...,l_n\}$. Thus, $\sigma$ may be described by
	\( [l_1 \mapsto \sigma(l_1), \ldots,  l_n \mapsto \sigma(l_n)]\).

	\noindent
	Given two substitutions $\sigma$ and $\sigmahat$, the \emph{extension} of $\sigma$ with $\sigmahat$ is defined as
	$\twosubs{\sigma}{\sigmahat}\colon \domain{\sigma} \cup \domain{\sigmahat} \to \BT$ such that
	\[\begin{array}{cl}
	
		\twosigs(l) &= \begin{cases} \sigma(l) &\textrm{if } l \in \domain{\sigma},\\ 
				                     \sigmahat(l) &\textrm{otherwise.}
					   \end{cases} \\
	\end{array}\]
\end{defn}

\begin{defn}\label{def:solution}
	A substitution $\sigma$ is a \emph{solution} to $C$ if for all $c \in C$ it holds:
			%\item $\labelSet{C} \subseteq \domain{\sigma}$, 
 	\[ c \equiv \ B_1 \nlt B_2 \ \implies \vdash \sigma(B_1) \nlt \sigma(B_2)\]
	We denote this by $\sigma \vdash C$.  
	The substitution $\sigma_c$ is a \textit{minimum solution} to $C$, denoted by $\sigma_c \vdash_{min} C$, when
	\begin{itemize}
		\item $\sigma_c \vdash C$, 
		\item $\domain{\sigma_c} = \labelSet{C}$, 
		\item $\sigma \vdash C, l \in \labelSet{C} \implies \vdash \sigma_c(l) \nlt \sigma(l)$.\qed
	\end{itemize}
\end{defn}

\begin{lemma}[Uniqueness of minimal solution]\label{lemma:unique-minimal}  
	\[ \forall C,\sigma_1, \sigma_2. \qquad \sigma_1, \sigma_2 \vdash_{min} C \implies \sigma_1 = \sigma_2. \]
\end{lemma}
\begin{proof}
	We readily have that $\domain{\sigma_1} = \labelSet{C} = \domain{\sigma_2}$.
	As both
	$\vdash \sigma_1(l) \nlt \sigma_2(l)$ and $\vdash \sigma_2(l) \nlt \sigma_1(l)$ hold
	for all $l \in \labelSet{C}$, the equality $\sigma_1(l) = \sigma_2(l)$ follows from definition of \emph{Satisfied Constraint}.
\end{proof}

\begin{lemma}[Existence of solutions]\label{lemma:cnf-minimal}
	\[ C \in \CNF \implies \exists ! \sigma_c. \  \sigma_c \vdash_{min} C. \]
\end{lemma}
\begin{proof}
	Consider $C \in \CNF$. Define a substitution $\sigma_c$ as $[l_1 \mapsto \stat, \ldots, l_n \mapsto \stat]$, where $\labelSet{C} = \set{l_1, \ldots, \l_n}$.
	As $\sigma_c$ clearly solves all constraints of the form $l \nlt \tilde{l}, \tb{S} \nlt l \ or \ \tb{D} \nlt l$ with $l, \tilde{l} \in \CNF$, it is a solution
	$\sigma_c \vdash C$.\\
	The minimality follows from $\domain{\sigma_c} = \labelSet{C}$ and from the fact that both $\vdash \stat \nlt \stat$ and $\vdash \stat \nlt \dyn$.
\end{proof}

\subsection{Normal Form and Normalization}
\begin{wrapfigure}{h!}{0.4\textwidth}
%\centering
\footnotesize
\[
\arraycolsep=1pt\def\arraystretch{1.4}
\begin{array}{l|l|l|l}
    \multicolumn{1}{c	}{Redex}&  & \multicolumn{2}{c}{Result}  \\
\hline
\   C \in \CSets 					   &\  & \ \sigma 		      &\tilde{C}\\
\hline
a)~ C_0 \cup \{ \tb{S} \nlt \tb{S} \} &\to& [\ ] 				  &C_0\\
b)~ C_0 \cup \{ \tb{S} \nlt \tb{D} \} &\to& [\ ] 				  &C_0\\
c)~ C_0 \cup \{ \tb{D} \nlt \tb{D} \} &\to& [\ ] 				  &C_0\\
d)~ C_0 \cup \{ l \nlt \tb{S} \}          &\to& [l \mapsto \tb{S}] &[l \mapsto \tb{S}](C_0)\\
e)~ C_0 \cup \{\tb{D} \nlt l\}            &\to& [l \mapsto \tb{D}] &[l \mapsto \tb{D}](C_0)
\end{array}
\]
 \caption{Constraints Normalization ($\to$)} \label{fig:rewrite-rules}
\end{wrapfigure}
Lemma~\ref{lemma:cnf-minimal} shows how the \textit{minimum solution} for a normal form constraint set can be found.  This section presents a set of rewrite rules that transform any constraint set to the corresponding normal form, thus making the solution easy to find. In Fig.~\ref{fig:rewrite-rules}, a set of normalizing rewrite rules are shown.  
\begin{defn}\label{def:application-rewrite}
	For a $C \in \CSets$ the \emph{application} of a normalization rewrite rule from Fig.~\ref{fig:rewrite-rules} returns a constraint set $\tilde{C}$ and a substitution $\sigma$. 
	We denote this by $C \to^{\sigma} \tilde{C}$.  Exhaustive application is denoted by $C \to^{{}_\star\  \sigma} \tilde{C_k}$.
	\qed
\end{defn}

\begin{lemma}[Termination]\label{lemma:fixpoint} For all $C$, $\tilde{C}$, and $\sigma$, whenever $C \to^{\sigma} \tilde{C}$ then $|\tilde{C}| < |C|$.
\end{lemma}
\begin{proof}
	By inspecting Fig.~\ref{fig:rewrite-rules}, it is obvious that every rule reduces the number of constraints by one.
\end{proof} 

\begin{lemma}[Solution preservation]\label{lemma:composition}
	\[
		C \to^{\sigma} \tilde{C} \ \wedge \ \tilde{\sigma} \vdash \tilde{C} \implies \sigma \se \tilde{\sigma} \vdash C.
	\]
\end{lemma}

\begin{proof}
	By case analysis of the normalization rewrite rules.
	a), b), c): Since $\sigma$ is the empty substitution, it is clear that $\sigma \se \tilde{\sigma} = \tilde{\sigma}$. 
	As $\labelSet{C} = \labelSet{\tilde{C}}$, it follows that $\sigma \se \tilde{\sigma} \vdash C$.
	
	\noindent
	d), e): The constraint removed from $C$ is solved by $\sigma$. Thus, $\tilde{\sigma} \vdash \sigma(C)$ follows from $\tilde{\sigma} \vdash \tilde{C}$. 
	As $\labelSet{\sigma(C)} = \labelSet{\tilde{C}}$ and $\domain{\sigma} \cap \labelSet{\tilde{C}} = \varnothing$,
	we obtain $\sigma \se \tilde{\sigma} \vdash C$.
\end{proof}

\begin{lemma}[$\CNF$ or $\CError$] For all $C$ and $\tilde{C}$,  whenever $C \to^{*\ \sigma}  \tilde{C}$ then
	\[\tilde{C} \in \CNF  \ \vee \   \tilde{C} \in \CError.\]
\end{lemma}

\begin{proof}
 	By definition of \emph{exhaustive application}, no rewrite rule can apply to $\tilde{C}$, then only four types of 	constraints may appear in $\tilde{C}$.
	Namely $\stat \nlt l$, $l \nlt \dyn$, $l \nlt \tilde{l}$ or $\dyn \nlt \stat$. 
	If $\dyn \nlt \stat$ is present, then $\tilde{C} \in \CError$, else $\tilde{C} \in \CNF$.
\end{proof}

\begin{lemma}[Minimal solution preservation]\label{lemma:composition-minimal}
	\[C \to^{\sigma} \tilde{C} \ \wedge \ \sigma_{\tilde{c}} \vdash_{min} \tilde{C} \implies \sigma \se \sigma_{\tilde{c}} \vdash_{min} C  .\]
\end{lemma}

\begin{proof}
	Lemma~\ref{lemma:composition} implies that $\sigma_c = \sigma \se \sigma_{\tilde{c}} \vdash C$.
	To show minimality first note that $\domain{\sigma_c} = \domain{\sigma} \cup \labelSet{\tilde{C}} = \labelSet{C}$ with the union being disjoint.
	Now assume that $l \in \labelSet{C}$, $\tilde{\sigma} \vdash C$ and consider the following two cases.
	
	\noindent
	$l \in \labelSet{\tilde{C}}$: Clearly $\tilde{\sigma} \vdash \tilde{C}$. As $\sigma_{\tilde{c}}$ is minimal, $\vdash \sigma_{\tilde{c}}(l) \nlt \tilde{\sigma}(l)$. 
	Thus, from $\sigma_c(l) = \sigma_{\tilde{c}}(l)$ we get that $\vdash \sigma_c(l) \nlt \tilde{\sigma}(l)$. 
	
	\noindent
	$l \in \domain{\sigma}$: $\sigma$ was obtained by applying rule e) or d). Thus, either $\sigma_c(l) = \sigma(l) = \stat$ or  
	$\sigma_c(l) = \sigma(l) = \dyn = \tilde{\sigma}(l)$. In both cases $\vdash \sigma_c(l) \nlt \tilde{\sigma}(l)$ holds.
\end{proof}

\begin{theorem}[Unique minimal solution]\label{thm:minimal}
	\[
		C \to^{*\ \sigma} \tilde{C} \ \wedge \ \tilde{C} \in \CNF \implies \exists! \sigma_c.~\sigma_c \vdash_{min} C .
	\]
\end{theorem}
\begin{proof}
	We obtain $\sigma_{\tilde{c}}$ such that $\sigma_{\tilde{c}} \vdash_{min} \tilde{C}$ from Lemma~\ref{lemma:cnf-minimal}. 
	By definition we have that $\exists k \in \naturals$ and $\sigma_1, \ldots, \sigma_k$ such that 
	$C \to^{\sigma_1} \tilde{C_1} \to^{\sigma_2} \ldots \to^{\sigma_k} \tilde{C}_k = \tilde{C}$ and $\sigma = \sigma_1 \se \ldots \se \sigma_k$.
	Thus, by introducing $\tilde{C}_0 = C$, we get $\sigma_i \se \ldots \se \sigma_k \se \sigma_{\tilde{c}} \vdash_{min} \tilde{C}_{i-1}$ 
	for all $1 \leq i \leq k$ from Lemma~\ref{lemma:composition-minimal}. Thus, $\sigma_c \vdash_{min} C$ with $\sigma_c = \sigma \se \sigma_{\tilde{c}}$.
\end{proof}
Combine Theorem~\ref{thm:minimal} and Lemma~\ref{lemma:unique-minimal},  the \textit{minimal solution} to any constraint set can be found as follows: first normalize the constraint set by applying constraint rewrite rules in figure 3, and then find the unique minimal solution to the normal form constraint set. The composition of the substitutions is the \textit{minimal solution}.
%  The least upper bound of a set S is the smallest b such that for all s in S, s <= b.
\subsection{Binding Time Analysis Correctness}
\begin{lemma} [Completeness]Consider a binding time type environment such that $\btenv(x_i) = \type_i^{b_i}$, and  $\rho(\scope,x_i) \in \text{Labels}$.  Then,
$\forall e,s,b$:
	\begin{itemize}
		\item $  \btenv  \vdash e:\type^b$ $\implies $ $\exists !\sigma.  \sigma\vdash_{min} |[e|]_{\scope, l, \rho}$ $\wedge \ \sigma(l) \nlt b$ 
		\item $\btenv  \vdash s^b$ $\implies$ $\exists!\sigma. \sigma \vdash_{min} |[s|]_{\scope,l,\rho}$ $\wedge\  \sigma(l) \nlt b$
	\end{itemize}
 \end{lemma}
\begin{lemma} [Soundness] \label{lemma:faithful}
Consider a typing environment such that $\tenv(x_i) = \type_i$, and  $\rho(\scope,x_i) \in \text{Labels}$.  There exists a binding time environment 
such that $|\btenv| =  \tenv$ and 
 	\begin{itemize}
		\item $\forall \ e .$  $\tenv \vdash e:\type$ $\wedge$ $\sigma \vdash_{min} |[e|]_{\scope, l, \rho}$ $\implies$ $\btenv \vdash e:\type^{\sigma(l)}$
		\item $\forall \ s.$  $\tenv \vdash s$ $\wedge$ $\sigma \vdash_{min} |[s|]_{\scope, l, 
		\rho}$ $\implies$ $\btenv \vdash s^{\sigma(l)}$
 	\end{itemize}
  \end{lemma}
The proofs for the two lemmas above are in Section 2 of the supplementary document.
\section{Conclusions and Future work}\label{section:conclusion}
In this paper we showed how the basic hybrid ODE formalism can be extended to support certain types of partial derivatives and equational constraints. The treatment is generic and so can be applied to any hybrid systems reachability analysis, and has been implemented in the context of the Acumen modeling language.
Interesting future work includes investigation of more advanced type systems\cite{freshML, cristiano06} that can ensure that a dynamic value is a variable so that we can build a type system that is able to detect all possible run-time errors that can interfere with the elimination of partial derivatives.

%\nocite{*}
\bibliographystyle{eptcs}
\bibliography{generic}

\appendix

\section{Appendix}\label{appendix}

\begin{comment}
\subsection{A Two-link Robot Arm Example}
\begin{figure}
  \centering
  \includegraphics[width=0.6\columnwidth]{figures/Arm.pdf}
  \caption{A Two-Link Robot Arm}\label{fig:arm}
\end{figure}
Two link robot arm is a classic example of rigid body dynamics and multi-link manipulator. It has two rigid bodies that are connected together and moving in a plane.  A diagram of the model is shown in Fig. \ref{fig:arm} with its physical parameters.  The length of each link is $2l$, and the center of mass
$m$ with inertia $I$.

We choose $\theta_1$ and $\theta_2$ as the generalized coordinates of the system, thus the system has two degrees of freedom.  The kinetic energy and potential energy of the system are shown below:
\[ T=\frac{1}{2} (I+m\ell^2)(\dot{\theta_1}^2+\dot{\theta_2}^2) + m\ell^2\dot{\theta_1}\dot{\theta_2}cos(\theta_2 -\theta_1)\]
\[V= -mg\ell(3cos\theta_1+cos\theta_2) \]
   
Applying the Lagrange equation shown in Equation 1 with $q = (\theta_1, \theta_2)$, we have the dynamic equations of the system.  Fig.~\ref{fig:armcode} shows the Lagrange model and the \CoreDEL code.  
In the ASC\RNum{2}-based syntax, keyword $\texttt{forall}$ is used to introduce universal quantifiers that range over  family of equations.  This example illustrates that the Lagrange method for models with multiple state variables can be directly expressed in \CoreDEL.
\end{comment}

\begin{figure*}
\footnotesize
$
\begin{array}{| l | l |}
\multicolumn{1}{l}{\fbox{\text{Cam and Follower}}}   \\
\hline
&\\
x = (1.5-\frac{cos(\frac{\pi}{2} - \theta))}{2}) * (1 + \frac{cos2(\frac{\pi}{2} - \theta)}{5}) &
\verb   \texttt{x = (1.5-cos(pi/2-t)/2)*(1+cos(2*(pi/2-t)/5))}\\ 
&\\
\ddot{\theta} = 1  \quad v = \pder{x}{\theta} \  \dot{\theta} \quad a = \dot{v} &
\verb   \texttt{t'' = 1, v = x'[t]*t', a = (v)'} \\
&\\
\hline
%\vspace{0.3in} \\
\multicolumn{2}{c}{}\\ 
\multicolumn{2}{c}{}\\  
\multicolumn{2}{l}{\fbox{\text{Compass Gait Biped}}}   \\ 
\hline
&\\

q = [\theta_1, \theta_2] \ \ m_1 = 1  \ \ m_2 = 2 \ \  l = 1   &
\verb   \texttt{q = (t1,t2), m1 = 1, l = 1, m2 =2,}\\
&\\
\gamma = 0.044 \quad g = 9.8 & \texttt{r = 0.044, g = 9.8,}\\
&\\
x_1 = \frac{1}{2}lsin\theta_1 \ \ y_1 = \frac{1}{2}lcos\theta_1  \ \ &
\verb   \texttt{x1 = l/2*sin(t1),y1 =l/2*cos(t1),}\\
&\\
 x_2 = lsin\theta_2\ \ y_2 = lcos\theta_2 &
\verb   \texttt{x2 = l*sin(t2), y2 = l*cos(t2),}\\
&\\
 x_3 = x_2+\frac{l}{2}sin(\theta_2-\theta_1)\ \  y_3 = y_2-\frac{l}{2}cos(\theta_2-\theta_1)  &
\verb   \texttt{x3=x2+l/2*sin(t2-t1),y3=y2-l/2*cos(t2-t1),}\\
&\\

L = T -V  \quad guard = lsin\gamma(sin\theta_1 + sin(\theta_2 - \theta_1) )& \verb   \texttt{L=T-V,guard=l*sin(r)*(sin(t1)+sin(t2-t1)),}\\
&\\
T=\frac{1}{2}m_1(\dot{x_1}^2 + \dot{y_1}^2 +  \dot{x_3}^2 + \dot{y_3}^2)&
\verb  \texttt{T = 1/2*m1*((x1)'\^{}2+(y1)'\^{}2+(x3)'\^{}2+(y3)'\^{}2)}\\
&\\
\qquad + \frac{1}{2}m_2(\dot{x_2}^2 + \dot{y_2}^2) &
\verb  \texttt{\qquad + 1/2*m2*((x2)'\^{}2 + (y2)'\^{}2),}\\ 
&\\
V= m_1g(y_1+y_3) + m_2gy_2 & \verb   \texttt{V = m1*g*(y1+y3)+m2*g*y2,  }\\
&\\
 H = H_1^{-1}\cdot H_2\cdot[\dot{\theta_1}^{-},\dot{\theta_2}^{-}]^{T} &
 \texttt{H = inv(H1)*H2*trans((t1',t2')),}\\
 &\\
\forall i\in\{1...|q|)\}. \frac{d}{dt} \left(\frac{\partial L}{\partial \dot{q_i}}\right)
 - \frac{\partial L}{\partial q_i}=0 &  \texttt{foreach i in 0:length(q) - 1 do}\\
 &
     \texttt{L'[(q(i))']' - L'[q(i)] = 0 ,} \\
 &\\
 H_1 = & \texttt{H1 = }\\
   \text{[} m_1l^{2}(\frac{5}{4}-\frac{cos\theta_2^{-}}{2})+m_2l^2  \quad
			  	 \frac{m_1}{4}l^2(1-2cos\theta_2^{-}) &  
				 \texttt{((m1*l\^{}2*((5/4-cos(t2)/2)+m2*l\^{}2)}, \\
& \hspace{1.2in} \texttt{m1/4*l\^{}2*(1-2*cos(t2))),}\\
&\\
 \frac{m_1}{2}l^2cos\theta_2^{-} \hspace{1.4in} \frac{m_1}{4}l^2  \text{]} &  \texttt{( m1/2*l\^{}2*cos(t2), 
\hspace{0.6in} m1/4*l\^{}2))},\\ 
&\\
 H_2 = [-\frac{m_1}{4}l^2+(m_2l^2+m1l^2)cos\theta_2^{-} \quad
			  	 -\frac{m_1}{4}l^2 &
\texttt{H2 =}\\
&\texttt{((-m1/4*l\^{}2+(m2*l\^{}2+m1*l\^{}2)*cos(t2),-m1/4*l\^{}2),}\\
\qquad \qquad \frac{m_1}{4}l^2 \hspace{1.4in}0 \ ] 
		    & \verb  \texttt{\qquad\qquad( m1/4*l\^{}2, \qquad\qquad\qquad\qquad\qquad\qquad\qquad0),}  \\	   
&\\
 \text{if} \ \  guard <0 \wedge \dot{guard} < 0\ \ \text{then} & \verb  \texttt{\text{if} guard <0 \&\& (guard)' \text{then}}\\
  \quad \theta_1^{+} = \theta_2^{-} - \theta_1^{-} \quad \theta_2^{+} = -\theta_2^{-} \quad  
 & \texttt{\quad t1 += t2 - t1, t2 += -t2,}  \\
  \quad \dot{\theta_1}^{+} = H(0)\quad \dot{\theta_2}^{+} = H(1) & \texttt{\quad t1' += H(0), t2' += H(1) \text{noelse}}\\
  &\\
\hline
\end{array}
$
\caption{Two Examples Problems in Mathematical Notation and in Acumen Syntax}
\label{fig:armcode}
\end{figure*}

\begin{figure*}
\footnotesize
\centering
\begin{tabular}{ l >{$}c<{$} >{$}l<{$}  >{$}l<{$}  }
%\hspace{0.5in}\fbox{\text{Syntax}} &&\\
\hspace{0in}Value & \vb&::=&\const^{\stat} \ | \ e^{\dyn} \ | \  \sequence{v_j^{b_j}}\\

\hspace{0in}Normal Form Equation & \wb &::=&  (x^b = \vb)^b \ | \ ({x^{+}}^{b} = \vb)^b \ | \  (\vb = \vb)^b  \ | \ 
                                                 (\tb{if} \ \vb\ \tb{then} \ \wb \ \tb{else} \ \wb)^b
                                        		  \ | \ (\sets{w_j^{b_j}})^b \\
 \multicolumn{4}{l}{ $\hspace{0in}\text{Function application} 
 \hspace{0.26in} |[\f(\sequence{v_j^{\stat}}) |] = v^\stat \ \ \text{such that} \ \ v^\stat \equiv \f(\sequence{v_j})^\stat$}\\
\multicolumn{4}{l}{ $\hspace{0in}\text{Time derivative} \hspace{0.48in} \tf(\sequence{v_j^{\dyn}},\sequence{v_j'^{\dyn}})= v^\dyn \ \ \text{such that} \ \ v^{\dyn} \equiv (\der{\f(\sequence{v_j^{\dyn}},\sequence{v_j'^{\dyn}})})^\stat$}\\
 \multicolumn{4}{l}{$\hspace{0in}\text{Partial derivative} \hspace{0.42in} \pf(\sequence{v_j^{\dyn}},\sequence{v_j'^{\dyn}},x^\dyn) = v^\dyn \ \ \text{such that} \ \ v^{\dyn} \equiv \pder{\f(\sequence{v_j^{\dyn}},\sequence{v_j'^{\dyn}})^\dyn}{x}$}\\
 \multicolumn{4}{l}{\fbox{$\text{Free Variables}$} }\\
&&&\\
\multicolumn{4}{l}{ $FV(x^b) = \setof{x} \quad  FV(\sequence{{e_j}^{b_j}}) = \bigcup_j FV(e_j^{b_j})  \quad
			       FV((f(\sequence{e_j^{b_j}})^b)) =\bigcup_j FV(e_j^{b_j}) 
			       \ \ LV(x = e ) = \{x\}  $}\\
&&&\\			 
 \multicolumn{4}{l}{$FV(e_1^{b}(e_2^{b})^{b}) = FV(e_1^b) \cup FV(e_2^b) \ \  
 			      FV(\der(e^b)) = FV(e^b) \ \   FV(\pder{e_1^b}{e_2^b}) = FV(e_1^b) \cup FV(e_2^b) \ \  LV(\sets{s_j} ) = \bigcup_j LV(s_j)$} \\
\end{tabular}

\[
\begin{array}{ l  }
\multicolumn{1}{l}{\fbox{$e^{b} \ereduce \vb \cup \{\err\}$}}
\\
\infer{\const^{b} \ereduce \const^b}{} \quad
\infer{(\sequence{e_j^{b_j}})^{b} \ereduce (\sequence{\vtwo{j}{b_j}})^b}
	  {e_j^{b_j} \ereduce \vtwo{j}{b_j} }\quad
 \infer{e_1^{b}(e_2^{\stat})^{b} \ereduce v_i^{b}}
         {e_1^{b} \ereduce (\sequence{v_j^{b}})^b & e_2^{\stat} \ereduce i^{\stat}}\quad
\infer{\f(e^{\stat})^{\stat} \ereduce |[\f(\sequence{v_j})|] }{e^{\stat} \ereduce (\sequence{v_j^{\stat}})^\stat } \\
\infer{e_1^{b}(e_2^{\dyn})^{\dyn} \ereduce \vtwo{1}{b}(\vtwo{2}{\dyn})^{\dyn}}
         {e_1^{b} \ereduce \vtwo{1}{b} & e_2^{\dyn} \ereduce \vtwo{2}{\dyn}} \quad

\infer{x^{\dyn} \ereduce x^\dyn}{ }\quad
 \infer{\f(e^{\dyn})^{\dyn} \ereduce\f(v^{\dyn})^{\dyn} }
         {e^{\dyn} \ereduce v^{\dyn} } \quad
\infer{\der(e^{\stat})^{\stat} \ereduce 0^\stat}{e^{\stat} \ereduce q^\stat} \quad
\infer{\der(e^{b})^{b} \ereduce x'^b}{e^{b} \ereduce x^{b}} \quad
%\infer{\der(\sequence{v_j^{b_j}})^{b} \ereduce (\sequence{v_j'^{b_j}})^b}
%	{(\der{v_j^{b_j}})^{b_j} \ereduce {v_j'}^{b_j} & } \quad
\infer{( \pder{e_1^{b}}{e_2^{\dyn}} )^{\stat} \ereduce 0^\dyn}{e_1^{b} \ereduce q_1^{\dyn} & e_2^{\dyn} \ereduce x^{\stat} } \\
\infer{( \pder{e_1^{\dyn}}{e_2^{\dyn}} )^{\dyn} \ereduce i^\dyn}{e_1^{\dyn} \ereduce x_1^{\dyn} & e_2^{\dyn} \ereduce x_2^{\dyn} \\
	\multicolumn{2}{c}{ i^\dyn = \begin{cases} 1^{\dyn} \ \mbox{if} \ x_1 = x_2 \\ 0^{\dyn} \ \mbox{otherwise}\end{cases}}} 

\quad

\infer{\der(e^{\dyn})^{\dyn} \ereduce \tf(\vectoromit{v_j^{b_j}}^\dyn,\vectoromit{{v_j}'^{b_j}}^\dyn)^\dyn}
	{e^{\dyn} \ereduce f((\sequence{v_j^{b_j}})^\dyn)^\dyn & (\der{v_j^{b_j}})^{\dyn} \ereduce {v_j}'^{b_j}} \quad	  
\infer{\pder{e_1^{\dyn}}{e_2^{\dyn}} \ereduce  \pf(\vectoromit{v_j^{b_j}}^\dyn,\vectoromit{{v_j}'^{b_j}}^\dyn,x^\dyn)^\dyn}
   	{\multicolumn{3}{c}{e_1^{\dyn} \ereduce  f((\sequence{v_j^{b_j}})^\dyn)^\dyn} \\ e_2^{\dyn} \ereduce x^\dyn 
	& (\pder{v_j^{b_j}}{x^\dyn})^{b_j} \ereduce {v_j}'^{b_j} }
	
 \end{array}\]
\[
\begin{array}{ l  }
\multicolumn{1}{l}{\fbox{$s^{b} \sreduce \done \cup \{\err\}$}} \\ 
\\
  \infer{\{x^{b}=e^{b}\}^{b} \uplus \sbs \sreduce  \{x^b=\vb\} \uplus \dones}
  	    { FV(e^b)\cap LV(\sbs) = \emptyset & e^b \ereduce \vb & \sbs[x^b := \vb] \sreduce \dones} \quad
  \infer{\{{x^{+}}^{b_1}=e^{b}\}^{b} \uplus \sbs \sreduce  \{{x^{+}}^{b_1}=\vb\}^b \uplus \dones}
  	    {  FV(e^{b})\cap LV(\sbs) = \emptyset& e^b \ereduce \vb & \sbs \sreduce \dones} \\
\\
  \infer{\{e_1^{b_1}=e_2^{b_2}\}^{b} \uplus\sbs\sreduce  \{ \vtwo{1}{b_1}= \vtwo{2}{b_2}\}^b \uplus \dones}
  	    {  \multicolumn{3}{c}{FV(e_1^{b_1})\cup FV(e_2^{b_2})\cap LV(\sbs) = \emptyset} \\ e_1^{b_1} \ereduce \vtwo{1}{b_1} &
  	     e_2^{b_2} \ereduce \vtwo{2}{b_2}&\sbs\sreduce \dones} \quad \quad
 \infer{ \{\tb{if}\  e^{\stat} \ \tb{then}\  s_1^{b_1} \ \tb{else}\  s_2^{b_2}\}^b \uplus\sbs\sreduce w_j^{b_j} \uplus \dones}
 { \multicolumn{4}{l}{e^{\stat} \ereduce t_j^{\stat} \ \   s_j^{b_j} \sreduce w_j^{b_j} \ \ \sbs\sreduce \dones}\\
   FV(e^\stat)\cap LV(\sbs) = \emptyset &
 t_1^{\stat} = \tb{true}^{\stat} & t_2^{\stat} = \tb{false}^{\stat} }\\
 \\ 
 \qquad \qquad \qquad \infer{ \{\tb{if}\  e^{\dyn} \ \tb{then}\  s_1^{b_1} \ \tb{else}\  s_2^{b_2}\}^b \uplus\sbs\sreduce 
  			\{\tb{if}\  v^{\dyn} \ \tb{then}\  \dones_1 \ \tb{else}\  \dones_2\}^b \uplus \dones}
   {\multicolumn{3}{c}{FV(e^{\dyn})\cap LV(\sbs) = \emptyset \quad  e^{\dyn} \ereduce v^{\dyn}}\\ 
     s_1^{b_1} \sreduce \dones_1 & s_2^{b_2} \sreduce \dones_2&\sbs\sreduce \dones}\\ \\
     
\qquad \qquad \qquad \qquad    \infer{\{\forall n \in e^{b} \ s_1^{b_1}\}^{b_1} \uplus\sbs\sreduce  (\{w_j^{b_1}\}^{j \in 1..n})^{b_1} \uplus \dones}
      {\multicolumn{2}{c}{FV(e^b)\cap LV(\sbs) = \emptyset\ \  e^{b} \ereduce  (\sequence{\vtwo{j}{b}})^b} \\
       \multicolumn{2}{c}{ (\{s_1^{b_1}[n := v_j^{b}]\}^{j \in 1..m})^{b_1} \sreduce (\{w_j^{b_1}\}^{j \in 1..m})^{b_1}} \\ \sbs\sreduce \dones}\\

\end{array}
\]
% Error evluations

\centering\caption{Specialization big-step semantics}\label{fig:semantics}
\end{figure*}

\end{document}